\documentclass[12pt]{article}
\usepackage[mathscr]{eucal}
\usepackage{amsthm,amsmath,amssymb,amscd}
\usepackage{graphicx}
\usepackage{latexsym,enumerate}
\usepackage{color}
\usepackage{tikz}
\usetikzlibrary{shapes.arrows,chains,positioning}
\usepackage[outline]{contour}
\usepackage{caption}
\usepackage{subcaption}
\newtheorem{theorem}{Theorem}
\newtheorem{lemma}{Lemma}
\newtheorem{proposition}{Proposition}

\newtheorem{conjecture}{Conjecture}
\theoremstyle{definition}
\newtheorem{definition}{Definition}

\newtheorem*{remark}{Remark}

\usepackage{color,hyperref}
\definecolor{dark-blue}{rgb}{0.15,0.15,0.4}
\definecolor{dark-red}{rgb}{0.4,0.15,0.15}
\definecolor{medium-red}{rgb}{0.6,0,0}
\definecolor{medium-blue}{rgb}{0,0,0.6}
\hypersetup{
    colorlinks=true,
    linktoc=all,
    citecolor=medium-red,
    filecolor=medium-blue,
    linkcolor=medium-blue,
    urlcolor=medium-blue,
}

\begin{document}
\title{On the uniqueness of co-circular four body central configurations}
\author{Manuele Santoprete\thanks{ Department of Mathematics, Wilfrid Laurier
University E-mail: msantopr@wlu.ca}} 
 \maketitle

\begin{abstract}
    We study central configurations lying on a common circle in the Newtonian
four-body problem. Using  a topological argument  we prove that there is at most one co-circular central configuration for each cyclic ordering of the  masses on the circle. 

  \end{abstract}

\renewcommand{\thefootnote}{\alph{footnote})}
\tableofcontents

\section{Introduction}%
The Newtonian $n$-body problem is the study of the dynamics of $n$ point
particles with positive masses, moving according to Newton's laws of motion.
A {\it central configuration} (c.c.) of the $n $-body problem is a configuration of 
$n$ bodies where the  acceleration vector of each body is a common scalar multiple of its position vector with respect to the center of mass. 
The study of central configurations  in the Newtonian $n$-body problem  has a long history dating back to Euler and Lagrange,  and has become an active  sub-field of celestial mechanics.
While the  relative equilibria of the three-body problem have long been known, a complete classification is not  known for $ n>3 $. Even the finiteness of central configurations is a hard problem and it was only established in  the four-body problem by Hampton and Moeckel \cite{hampton2006finiteness} and in the five-body problem  (except for masses in a codimension two subvariety)  by Albouy and Kaloshin \cite{albouy2012finiteness}. A related problem is the study of central configurations for point vortices. A  classification was obtained in the four-vortex problem in the  case some of the vorticities are equal \cite{hampton2014relative,perez2015symmetric}.

In this paper we focus on  a subset of the four-body convex central configurations. A configuration is convex if no body lies inside or on  the convex hull formed by the other three bodies. MacMillan and Bartky \cite{macmillan1932permanent} proved that for any four positive masses
and any assigned order, there is at least one  convex planar central configuration of the
4-body problem with that order. See also   Xia \cite{xia2004convex} and Moeckel \cite{Moeckel2015} for  simpler proofs. Yoccoz \cite{Yoccoz1986} conjectured that there is only one such  configuration. 
\begin{conjecture}[Sim\'o-Yoccoz]\label{conj:convex}  There is a unique  convex planar central configuration of the
4-body problem for each ordering of the masses in the boundary of its convex hull.
\end{conjecture}
This conjecture is also implicit  in Sim\'o's paper  \cite{simo1978relative} and, according to Alain Albouy, it is likely the result of several conversations between Sim\'o and Yoccoz, and thus it is reasonable to call it Sim\'o-Yoccoz conjecture.  
The conjecture was first published, as far as we know, by Albouy and Fu  \cite{albouy2007euler}, see also  (see also \cite{albouy2008symmetry,perez2007convex}), and  was  included in the well known list of open problems on the classical $n$-body problems compiled by Albouy, Cabral and Santos \cite{albouy2012some}.
Results related to this conjecture were obtained by either putting restrictions on the geometry or restrictions   on the masses. 
In 1932, MacMillan and Bartky \cite{macmillan1932permanent} already  proved uniqueness in the particular case of  isosceles trapezoid configurations with two pairs of equal masses  located at adjacent vertices of a trapezoid. Similar results were also obtained by Xie  \cite{xie2012isosceles}.
The conjecture is also  known to be true if all the masses are equal \cite{albouy1995symetrie,albouy1996symmetric}, if two pairs of masses are equal \cite{perez2007convex,albouy2008symmetry,fernandes2017convex}, and for the case of three small masses \cite{corbera2015bifurcation}. Some of these   results also   hold for homogeneous power-law potentials  \cite{albouy2008symmetry,Fernandes_2019}.

Our goal in this work is to prove  prove the conjecture for the four-body {\it co-circular central configurations} (c.c.c's), namely those four-body c.c's which lie on a common circle.  Specifically we prove the following theorem 
\begin{theorem} \label{thm:uniqueness}
    There is at most one co-circular  central configuration of four bodies for each cyclic ordering of the masses. 
\end{theorem}
It is known, however, that for most values of the masses  there are  no co-circular configuration \cite{cors2012four}.
Although the co-circular four body problem may seem somewhat far-fetched, it is hoped that this work will prove useful in understanding the conjecture for  general four-body convex configurations.
 Furthermore, the co-circular  problem  has already attracted some attention \cite{cors2012four,hampton2005co,hampton2016splendid,llibre2015co,deng2017some,albouy1995symetrie,alvarez2013co}, in part because  of the following  conjecture, proposed by Alain Chenciner in 2001 \cite{llibre2015co,albouy2012some}.

 \begin{conjecture}
 Is the regular $n$-gon with equal masses the unique central configuration such that all the bodies lie on a circle, and the center of mass coincides with the center of the circle? 
 \end{conjecture}

Note that requiring the bodies to lie on a common circle  effectively restricts each body to  a one dimensional manifold, allowing for a more straightforward investigation. This situation is reminiscent of the collinear $ n$-body problem, where the bodies lie on a straight line. 

The first result concerning uniqueness of central configurations was obtained by Moulton for the collinear problem \cite{moulton1910straight}. Moulton  proved that there is a unique collinear central configuration for each ordering of the masses on the line.  A topological  proof of this fact  that uses Morse theory was then provided by Smale \cite{smale1970topologyII}. The approach to uniqueness we take in this paper is inspired by the topological approach of Smale. 

We will show that, in the four body co-circular problem,  the critical points  of the potential  restricted to a certain subset  (that will be described in a subsequent section) are also local minima of the potential.
This result is general and not restricted to the case where the center of mass coincides with the center of the circle. 
The key idea here is to use Ptolemy's theorem to characterize   co-circular configuration as done by  Cors and Roberts \cite{cors2012four}. Once we know that the critical points are local minima,  Morse theory can be used to prove  Theorem \ref{thm:uniqueness}.

The paper is organized as follows. In Section 2 we introduce the n-body problem and define central configurations. In Section 3 we write central configurations in terms of mutual distances between the bodies.
In Section 4 we define co-circular configurations and find their equations following the approach of Cors and Roberts \cite{cors2012four}. In particular we view such configurations as critical points of the potential restricted to a certain space that we call $ \mathcal{M} ^{ + } $.  
In Section 5 we prove Theorem \ref{thm:uniqueness} using Morse theory. This is done in four steps. In Proposition \ref{prop:crit-points} we show that all the critical points are nondegenerate local minima. 
In Lemma \ref{lem:M+} we show that the space $ \mathcal{M} ^{ + } $ is contractible and we  obtain  its Euler characteristic. In Lemma \ref{lem:uniqueness} we use Morse theory and the Euler characteristic of  $ \mathcal{M} ^{ + } $ to prove that the potential restricted to $ \mathcal{M} ^{ + } $ has a unique critical point. We then use this  last result to prove  Theorem \ref{thm:uniqueness}.

\section{Central Configurations of the Newtonian $n$-body problem}
Let $ P _1,P _2, P _3, \ldots ,P _n $ be   $n$  points in $\mathbb{R}^d$ with position vectors $\mathbf{q} _1 , \mathbf{q} _2 , \ldots ,\mathbf{q} _n  $. Let $ r _{ ij } = \| \mathbf{q} _i - \mathbf{q} _j \| $, be the Euclidean distance between the point $ P _i $ and $ P _j $, and let $\mathbf{r} = (r _{ 12 }, \ldots , r _{ n - 1 \,n }) $ be the vector of mutual distances.    
The  Newtonian $n$-body problem concerns the motion of $n$ particles
with positive masses   $ m _i>0 $ and positions $\mathbf{q} _i\in{\mathbb
R}^d$, where $i=1,\ldots,n$. .
The center of mass of the system is fixed at the origin of the coordinate systems, that is, we have $ \mathbf{q} _{ CM } = \frac{ 1 } { M} \sum _{ i = 1 } ^n m _i \mathbf{q} _i =0$, where $ M = m _1 + \ldots +m _n $ is the total mass. 
The motion is governed by Newton's
law of motion 
\begin{equation} 
 m _i\mathbf{\ddot q} _i=   \sum _{ i \neq j } \frac{m _i m _j (\mathbf{q} _j - \mathbf{q} _i)   } { r _{ ij } ^3 }=\frac{\partial  \tilde U}{  \partial  \mathbf{q}  _i}, \quad 1\leq i\leq n
\end{equation} 
where $\tilde U(\mathbf{q} )$ is the Newtonian potential
\begin{equation} 
    \tilde U( \mathbf{q} )=\sum_{i<j}\frac{m_im_j}{\| \mathbf{q} _i - \mathbf{q} _j \| },\quad 1\leq i\leq n.
\end{equation} 
Let us denote the Newtonian potential by $ U (\mathbf{r}) $ when viewed as a function of $ \mathbf{r} $.  
 A   {\it central configuration} (c.c.) of the $n$-body problem is a configuration $ \mathbf{q} \in \mathbb{R} ^{ nd} $ which satisfies the algebraic equations
 \begin{equation} \label{eqn:cc1}
 \lambda\, m _i \mathbf{q} _i   = \sum _{ i \neq j } \frac{ m _i m _j (\mathbf{q} _j - \mathbf{q} _i) } { r _{ ij } ^3 }, \quad 1 \leq i \leq n .
 \end{equation}
The central configuration equation  \eqref{eqn:cc1} is invariant under rotations, reflections and dilations.
It is standard to  say that two configurations $ \mathbf{q}  $ and $ \mathbf{q}' $ are {\it equivalent } if there is a non-zero constant $ k \in \mathbb{R}  $ and an orthogonal matrix $ R $ such that $ \mathbf{q} _i ' = k R \mathbf{q} _i $, $ i = 1, \ldots n $. This defines an equivalence relation $\sim'$, and thus  one can speak of  {\it equivalence classes of central configurations}. By convention,  when counting central configurations it is standard to  count the number of equivalence classes with respect to the equivalence relation $\sim'$. This convention is also used in the statement of Conjecture \ref{conj:convex} and of Theorem \ref{thm:uniqueness}.

Let $\tilde  I (\mathbf{q}) $ denote the moment of inertia, that is,
\[ \tilde I (\mathbf{q}) = \frac{1}{2} \sum _{ i = 1 } ^n m _i \| \mathbf{q} _i  \| ^2 
\]
and let $ I (\mathbf{r})  = \frac{ 1}{2 M} \sum _{ 1 \leq i < j \leq n } ^n  m _i m _j r _{ ij } ^2 $ denote the moment of inertia when viewed as a function of $ \mathbf{r} $. 

Using the moment of inertia we can write equation \eqref{eqn:cc1} as  
\begin{equation}\label{eqn:cc2}  \nabla _{ \mathbf{q} }  \tilde U (\mathbf{q}) = \lambda\, \nabla _{ \mathbf{q} }  \tilde I (\mathbf{q}),
    \end{equation}
where $ \nabla _{ \mathbf{q} } = 
\left[ \frac{ \partial } { \partial \mathbf{q} _1 } , \ldots , \frac{ \partial } { \partial \mathbf{q} _n }
\right] ^{ T }   $.
   Viewing $ \lambda $  as a Lagrange multiplier, we have that $ \mathbf{q} $ is  a central configuration if and only if it is a critical point  (with respect to $ \mathbf{q} $) of $\tilde U(\mathbf{q} ) $    subject to the constraint $ \tilde I = \tilde I _0 $. 

For any configuration  $\mathbf{q} $ the vectors  $ \mathbf{q} _j $  span a subspace $ \mathcal{C} (\mathbf{q}) $   of $\mathbb{R}  ^{ d } $  called the {\it  centered position space} \cite{Moeckel2015}. It is
natural to define the dimension of a configuration to be $ \operatorname{dim} (
\mathbf{q} ) = \operatorname{dim} (\mathcal{C} (\mathbf{q})) $.
We say that $\mathbf{q}$ is a  {\it Dziobek configurations} if $ \operatorname{dim } (\mathbf{q}) = n - 2 $ \cite{Moeckel2015}. 
In the four-body problem  $\mathbf{q}$ is a  {\it Dziobek central configuration}  if it is a central configuration with $ \operatorname{dim } (\mathbf{q}) = 2 $. The set of  four-body Dziobek configurations coincides with the set of planar, non-collinear, central configurations.

\section{Central  Configurations in terms of distances}
For four bodies it is convenient to  recast the equations defining Dziobek central configuration, so that
the variables are the distances between the particles rather than their coordinates.
 Since the mutual distances determine the configuration up to rotation and reflection symmetry, this choice not only reduces the number of variables but also  removes  the rotational  and reflectional degeneracy. The dilational degeneracy can then be eliminated by fixing the size of the configuration with the restriction $ I = 1 $. 

Let $ \mathbf{r}= (r _{ 12 } , r _{ 13 } , r _{ 14 } , r _{ 23 } , r _{ 24 } , r _{ 34 }) \in (\mathbb{R}  ^{ + }) ^{ 6 } $ be a vector of  non-negative mutual distances, and let  the  Cayley--Menger determinant of four points  $ P _1 , \ldots P _4 $   be 
\[   H (\mathbf{r} ) = 288 V ^2 =  \begin{vmatrix}
            0 & 1 & 1 & 1 & 1 \\
            1 & 0 & r^2 _{ 12 } & r^2 _{ 13 } & r^2 _{ 14 }  \\
            1 & r ^2_{ 12 } & 0 & r^2 _{ 23 } & r^2 _{ 24 } \\
            1 & r ^2_{ 13 } & r^2 _{ 23 } & 0 & r^2 _{ 34 } \\
            1 & r^2 _{ 14 } & r^2 _{ 24 } & r^2 _{ 34 } & 0
        \end{vmatrix}. 
\]
where $ V $ is the volume of the configuration. 
Not all vectors $ \mathbf{r} $  realize  actual configurations of four bodies in  $ 
\mathbb{R}  ^3 $ (see \cite{cors2012four} for some examples). A necessary and sufficient condition for a given vector $\mathbf{r}$ to correspond to an actual configuration of four bodies is that $ H (\mathbf{r}) \geq 0 $ and all strict triangle inequalities be satisfied.  As a consequence we  consider the sets 
\[ \mathcal{G} = \{\mathbf{r} \in (\mathbb{R}  ^{ + }) ^6 |\, H (\mathbf{r}) \geq 0 \mbox{ and }  r _{ ij } + r _{ jk } 
    >  r _{ ik } \mbox{ for all } (i, j, k) \mbox{ where } i \neq j \neq k  \}.
\]
and 
\[\mathcal{N} = \{ \mathbf{r} \in \mathcal{G} |\, I (\mathbf{r}) - 1 = 0 , \quad H (\mathbf{r}) = 0 \}   \]
We say that a  vector of mutual distances  $\mathbf{r}$   is {\it geometrically realizable}  if $ \mathbf{r} \in \mathcal{G} $ and that $\mathbf{r}$   is a {\it normalized Dziobek configuration}  if $ \mathbf{r} \in \mathcal{N} $. 

Thus we have the following characterization of planar four body central configurations given by Dziobek:

\begin{proposition} 

Let $ \mathbf{q}  $ be a  Dziobek configuration, let  $ \mathbf{r} \in \mathcal{N}  $ be  its corresponding normalized Dziobek configuration, and let  $U  |_{ \mathcal{N} }: \mathcal{N} \to \mathbb{R}   $ be  the restriction of the Newtonian potential $ U $   to $ \mathcal{N} $.  Then, $\mathbf{q}$ is a Dziobek central configuration if and only if $ \mathbf{r} $ is a critical point of  $U  |_{ \mathcal{N} } $ with respect to $ \mathbf{r} $.
\end{proposition}
Since equations \eqref{eqn:cc1} are invariant under rotations,  dilations and reflections in the plane, we can consider two relative equilibria as equivalent if they are related by these symmetry operations. This  defines an equivalence relation $\sim$, different from the more standard one introduced in section 2.
Let $X$  be the set of equivalence classes with respect to $\sim$,  then the set of equivalence classes  
$ X $ is in a one-to-one correspondence with the set  $ c (U|_{\mathcal{N}})$ of critical points of the function $ U (\mathbf{r})  |_{\mathcal{N}} $.




\section{Co-circular Central Configurations}
In this section we study co-circular central configurations. 
For a planar configuration we  say that the bodies 
are {\it ordered counterclockwise} ({\it clockwise}) if they are numbered consecutively while traversing the boundary of the quadrilateral in a countercklockwise (clockwise) direction. 
Since we use mutual distances as coordinates, we cannot distinguish between bodies  ordered counterclockwise and bodies ordered  clockwise. Hence, we introduce the following terminology: we say that the bodies are {\it ordered sequentially }   if they are numbered consecutively while traversing the boundary of the quadrilateral in any direction.  

 Without loss of generality, we may assume that any cyclic quadrilateral is ordered sequentially so that $ r _{ 13 } $ and $ r _{ 24 } $ are the lengths of the diagonals. This is justified because we can always relabel the bodies so that they are ordered sequentially.
Denote
\[ P (\mathbf{r}) = r _{ 12 } r _{ 34 } + r _{ 14 } r _{ 23 } - r _{ 13 } r _{ 24 }.  \]
Ptolemy's theorem states that if a quadrilateral is  sequentially ordered  and cyclic then   $ P = 0 $.  More  in general,  Ptolemy's inequality says   that $ P \geq 0 $  for any   convex quadrilateral ordered sequentially and  for any  tetrahedron \cite{apostol1967ptolemy}. Equality holds if and only if the four bodies are co-circular. 
Let  $ \mathcal{P} $ be the set of geometrically realizable $\mathbf{r}$ satisfying $ P (\mathbf{r}) = 0 $, that is 
\[\mathcal{P} = \{ \mathbf{r} \in \mathcal{G} |\, P (\mathbf{r}) = 0 \}, \]
let $ \mathcal{M} $ be the set of $ \mathbf{r} $ which satisfy $ I = 1 $ and $ P = 0 $ 
\[ \mathcal{M} = \{ \mathbf{r} \in \mathbb{R}  ^6 | I (\mathbf{r})   - 1 = 0, \quad P (\mathbf{r}) = 0
    \}, \]
and let $ \mathcal{M} ^{ + } $ be the set obtained from $ \mathcal{M} $ by reintroducing the restrictions on the mutual distances 
\[ \mathcal{M}^+ = \{ \mathbf{r} \in (\mathbb{R}^+)  ^6 | I (\mathbf{r})   - 1 = 0, \quad P (\mathbf{r}) = 0 
    \}. \]
It is easy to see that $ \mathcal{M} $ is a smooth four-dimensional manifold, since the gradients of  
$ I $ and $ P $ are independent, and that $ \mathcal{M} ^{ + } $ is a manifold with boundary. 
In a later section we will show that $ \mathcal{M}\approx S ^2 \times S ^2 $ and that $ \mathcal{M} ^{ + } $ is homeomorphic to a closed ball. We denote the boundary  of $ \mathcal{M} ^{ + } $  by  $ \partial \mathcal{M} ^{ + } $ and observe that $U|_{\partial \mathcal{M} + } = \infty $. 

One last set that will play an important role in this paper is  $ \mathcal{D} $, which is  defined as follows:
\begin{align*} \mathcal{D}  & =  \{ \mathbf{r} \in  \mathcal{G }| I( \mathbf{r}) - 1 = 0 , P (\mathbf{r}) = 0 ,    H (\mathbf{r}) =0  \} \\ & = \{ \mathbf{r} \in \mathcal{M} ^{ + } \cap \mathcal{G }| H (\mathbf{r}) =0  \} .
\end{align*} 

There is an interesting relationship between the conditions $ P (\mathbf{r}) = 0 $ and  $ H (\mathbf{r}) = 0 $. The following Lemma sheds some light on this relationship, further insight is given in Lemma \ref{lem:grad}. 
\begin{lemma}\label{lem:h} If $ \mathbf{r} \in \mathcal{P} $, then $ H (r) = 0 $.  In other words on the set of geometrically realizable vectors for which $ P = 0 $   the configuration of four bodies  is coplanar. Moreover, we have that $ \mathcal{D} = \mathcal{M} ^{ + } \cap \mathcal{G} $.      
\end{lemma}
  \begin{proof}
Pech showed \cite{10.1007/978-3-642-21898-9_34} that the Cayley-Menger determinant can be written as follows 
\[\frac{1}{2} H (\mathbf{r}) =  P (\mathbf{r}) \cdot  Q (\mathbf{r})   -K ^2 (\mathbf{r})   \]
where 
\begin{align*} Q (\mathbf{r})   =&  r _{ 12 } r _{ 34 } (- r _{ 12 } ^2 - r _{ 34 } ^2 + r _{ 23 } ^2 + r _{ 14 } ^2 + r _{ 13 } ^2 + r _{ 24 } ^2) \\ &  + r _{ 14 } r _{ 23 } (r _{ 12 } ^2 r _{ 34 } ^2 - r _{ 23 } ^2 - r _{ 14 } ^2 + r _{ 13 } ^2+ r _{ 24 } ^2) \\& - r _{ 13 } r _{ 24 } (r _{ 12 } ^2 + r _{ 34 } ^2 + r _{ 23 } ^2 + r _{ 14 } ^2 - r _{ 13 } ^2 - r _{ 24 } ^2  ) 
\end{align*} 
and 
\[K (\mathbf{r})  = r_{12} r_{13} r_{23} - r_{12} r_{14} r_{24} + r_{13} r_{14} r_{34} - r_{23} r_{24} r_{34}.\]
If $ P  = 0 $, then 
\[
    \frac{1}{2} H (\mathbf{r}) = - (K (\mathbf{r})) ^2 \leq 0.   
\]
Since $ \mathbf{r} \in \mathcal{G} $ implies $ H (\mathbf{r}) \geq 0 $ it follows that $ K (\mathbf{r}) =  H (\mathbf{r}) = 0 $.   
\end{proof} 
A similar relationship  exists between $ H (\mathbf{r}) = 0 $ and  the condition
required of  four points to form a trapezoid.  This relationship was exploited in \cite{santoprete2018four} to obtain equations for trapezoidal central configurations. 

Since co-circular central configurations are Dziobek configuration we can give the following definition
\begin{definition}
   The configuration vector $ \mathbf{q} $ is a sequentially ordered cyclic  four-body  central configuration  if and only if its corresponding  distance vector $ 
   \mathbf{r} $   belongs to $ \mathcal{D} $ and it  is a critical point of  $ U|_{\mathcal{N}} $ with respect to $ \mathbf{r} $. 
  \end{definition} 
In terms of  Lagrange multipliers this means that $ \mathbf{r} \in \mathcal{D}  $ is a sequentially ordered cyclic four body central configuration if and only if it is a critical point of the function
\[
      U (\mathbf{r})   + \lambda M (I (\mathbf{r})   - 1) + \eta H (\mathbf{r}), 
\]
satisfying $ I - 1 = 0 $, $ P= 0 $ and $ H = 0 $, where $ \lambda $, and $ \eta $ are Lagrange multipliers.  The downside of this approach is that $ H $ and its derivatives are fairly complicated. Using the following lemma  however, it is possible to find simpler equations for the co-circular configurations. This  lemma was proven  in \cite{cors2012four}, here we provide a different proof. 
\begin{lemma} \label{lem:grad}
    For any $\mathbf{r} \in \mathcal{P} $
    \[
        \nabla _{ \mathbf{r} }  H (\mathbf{r}) =2 Q (\mathbf{r})  \nabla _{ \mathbf{r} } P  (\mathbf{r}) 
    \]
    where  $ \nabla _{ \mathbf{r} } = 
    \left[  \frac{ \partial } { \partial r _{ 12 } } , \ldots , \frac{ \partial } { \partial r _{ 34 } }
    \right] ^{ T }   $. In other words on the set of geometrically realizable vectors for which $ P $ vanish,  the gradients of $ H $ and $ P $ are parallel.   
\end{lemma}
\begin{proof} 
Since $ \frac{1}{2} H (\mathbf{r}) =  P (\mathbf{r}) \cdot  Q (\mathbf{r})   -K ^2 (\mathbf{r})   $ we have that 
\[
     \frac{1}{2}  \nabla _{ \mathbf{r} } H (\mathbf{r}) =    Q (\mathbf{r}) \nabla _{ \mathbf{r} } P  (\mathbf{r})  + P  (\mathbf{r})   \nabla _{ \mathbf{r} }  Q (\mathbf{r})  -2K  (\mathbf{r})  \nabla _{ \mathbf{r} }K (\mathbf{r}).
\]
Since $ \mathbf{r} \in \mathcal{P} $, then $ H = P = 0 $. It follows that $ K = 0 $ as well. Hence,
$  \frac{1}{2}  \nabla _{ \mathbf{r} } H (\mathbf{r}) =   Q (\mathbf{r})  \nabla _{ \mathbf{r} } P  (\mathbf{r}) $. 
\end{proof}
Comparing this result with the corresponding lemma in \cite{cors2012four} we find that if $ \mathbf{r} \in \mathcal{P}  $   then 
\[ 2 Q (\mathbf{r}) =  \left( \frac{ 4 } { r _c ^2} \prod_{i<j} r _{ ij } \right) \neq 0,  \]
where $ r _c $ is the circumradius of the cyclic quadrilateral.

We now have the following characterization of co-circular configurations due to Cors and Roberts (\cite{cors2012four}):
\begin{proposition}\label{prop:gradients}
    Let  $ \mathbf{r} \in \mathcal{D} $, then $\mathbf{r}$  is a critical point of  $  U|_{ \mathcal{N} } $, the restriction of $  U $ to $ \mathcal{N} $,  if and only if  $ \mathbf{r} $ is  a critical point of the function $  U| _{ \mathcal{M} ^{ + } }: \mathcal{M} ^{ + }  \to \mathbb{R}  $.
     Therefore the vector $ \mathbf{q} $ is a sequentially ordered cyclic four-body c.c. if and only if  the corresponding distance vector  $ \mathbf{r}    \in \mathcal{D}  $  is a critical point of the Lagrangian function
    \[L (\mathbf{r} ; \lambda, \sigma ) =  U (\mathbf{r})   + \lambda M \,(I (\mathbf{r})   - 1) + \sigma P (\mathbf{r})   \]
   satisfying $ I - 1 = 0 $, $ P= 0 $ and $ H = 0 $, where $ \lambda $ and $ \sigma $ are Lagrange multipliers.
\end{proposition}
\begin{proof}
 Recall that $\nabla _{ \mathbf{r} } U |_{\mathcal{M} ^{ + } }$ is the orthogonal projection of $ \nabla _{ \mathbf{r} } U (\mathbf{r})   $ onto  the tangent space $ T _{ \mathbf{r}}( \mathcal{M} ^{ + }) $, and similarly   $\nabla _{ \mathbf{r} } U |_{\mathcal{N} }$ is the orthogonal projection of $ \nabla _{ \mathbf{r} } U (\mathbf{r})   $ onto  the tangent space $ T _{ \mathbf{r}}( \mathcal{N} ) $, with  \[
     T _{ \mathbf{r} } \mathcal{N} = \{ \mathbf{v} \in   \mathbb{R}   ^{ 6 }\, |\, \nabla _{ \mathbf{r} } (I (\mathbf{r}) - 1 ) \cdot  \mathbf{v} = 0, \nabla _{ \mathbf{r} } H (\mathbf{r}) \cdot \mathbf{v}  = 0 \}  
 \]
and 
 \[
     T _{ \mathbf{r} } \mathcal{M} ^{ + } = \{ \mathbf{v} \in  \mathbb{R}   ^{ 6 } \, |\, \nabla _{ \mathbf{r} } (I (\mathbf{r}) - 1 ) \cdot  \mathbf{v} = 0, \nabla _{ \mathbf{r} } P (\mathbf{r}) \cdot \mathbf{v}  = 0 \}.
 \]
 Since $ \mathbf{r} \in \mathcal{D} $, by Lemma \ref{lem:grad}, $   \nabla _{ \mathbf{r} }  H (\mathbf{r}) =2 Q (\mathbf{r})  \nabla _{ \mathbf{r} } P  (\mathbf{r}) $. It follows that, if $ \mathbf{r} \in \mathcal{D} $, then 
$ T _{ \mathbf{r} } \mathcal{M} ^{ + } =  T _{ \mathbf{r} } \mathcal{N }  $, and hence 
 $\nabla _{ \mathbf{r} } U |_{\mathcal{N} }= \nabla _{ \mathbf{r} } U |_{\mathcal{M} ^{ + }  }$ for any $ \mathbf{r} \in \mathcal{D} $. 
Then   $\nabla _{ \mathbf{r} } U |_{\mathcal{M} ^{ + }  }=0 $ if and only if  $\nabla _{ \mathbf{r} } U |_{\mathcal{N}  }=0$, that is, $\mathbf{r}$  is a critical point of  $  U|_{ \mathcal{N} } $  if and only if  $ \mathbf{r} $  is  a critical point of the function $  U| _{ \mathcal{M} ^{ + } }$.

\end{proof} 
Proposition \ref{prop:gradients} says that if 
$ \mathbf{r} \in \mathcal{D} $, we can find the critical points of $ U|_{ \mathcal{N} } $ by finding the critical points of $ U|_{\mathcal{M} ^{ + } } $. 
The equations of the critical points of $ U| _{ \mathcal{M} ^{ + } } :\,{ \mathcal{M} ^{ + } } \to \mathbb{R} $, are given by $ \nabla _{ \mathbf{r} }  L (\mathbf{r} ; \lambda, \sigma ) = \nabla _{ \mathbf{r} }  U + \lambda M \nabla  _{ \mathbf{r} } I+ \sigma \nabla _{ \mathbf{r} }  P $, the gradient of the  Lagrangian $ L $. Explicitly, we have
\begin{align} 
m _1 m _2 ( r _{ 12 } ^{ - 3 } - \lambda)  & = \sigma\, \frac{ r _{ 34 } } { r _{ 12 } }   &   m _3 m _4 ( r _{ 34 } ^{ - 3 } - \lambda)  & = \sigma\, \frac{ r _{ 12 }}{ r _{ 34 }}  \label{eqn:cc_n1}\\ 
m _1 m _3 (r _{ 13 } ^{ - 3 } - \lambda) & = - \sigma   \frac{ r _{ 24 }}{ r _{ 13 }} &  m _2 m _4 (r _{ 24 } ^{ - 3 } - \lambda) & = - \sigma  \frac{  r _{ 13 } } {  r _{ 24 } }  \label{eqn:cc_n2}\\
m _1 m _4 (r _{ 14 } ^{ - 3 } - \lambda) & =  \sigma  \frac{  r _{ 23 } } {  r _{ 14 } } &  m _2 m _3 (r _{ 23 } ^{ - 3 } - \lambda) & =  \sigma \frac{   r _{ 14 }  } {r _{ 23 }}.\label{eqn:cc_n3}
\end{align}
It is important to observe that these equations hold for $ \mathbf{r} \in \mathcal{M} ^{ + } $, and not just for $ \mathbf{r} \in \mathcal{D} $, however, if $ \mathbf{r} \not\in \mathcal{D} $ then these equations do not give central configurations.  
Since $ \mathbf{r} \in \mathcal{M} ^{ + } $, the constraints  $ I - 1 = 0 $ and  $ P= 0 $ must be satisfied, but $ H = 0 $ is not required. When  $ \mathbf{r} \in \mathcal{D} \subset \mathcal{M} ^{ + }  $, however, $ H = 0 $ and the solutions of these  equations  give co-circular central configurations.  

The equations have been grouped in pairs so that when they are multiplied together the product of the right-hand sides is  $   \sigma ^2  $. This yields the  well-known  relation of Dziobek \cite{dziobek1900uber}
\begin{equation}\label{eqn:dziobek}
   (r _{ 12 }^{ - 3 }  - \lambda) (r _{ 34 } ^{ - 3 } - \lambda) = (r _{ 13 } ^{  - 3 } - \lambda) (r _{ 24 } ^{ - 3 } - \lambda) = (r _{ 14 } ^{ - 3 } - \lambda) (r _{ 23 } ^{ - 3 } - \lambda),        
\end{equation} 
which is required of any   4-body  Dziobek central configuration. From equations  (\ref{eqn:cc_n1}),(\ref{eqn:cc_n2}) and (\ref{eqn:cc_n3}) we  obtain three equations for $ \sigma ^2 $:
\begin{align} 
    \sigma ^2 & = m _1 m _2 m _3 m _4 (r _{ 12 } ^{ - 3 } - \lambda ) (r _{ 34 } ^{ - 3 } - \lambda)\label{eqn:sigma1}   \\
    \sigma ^2 & = m _1 m _2 m _3 m _4 (r _{ 14 } ^{ - 3 } - \lambda ) (r _{ 23 } ^{ - 3 } - \lambda)  \label{eqn:sigma2}  \\
    \sigma ^2 & = m _1 m _2 m _3 m _4 (r _{ 13 } ^{ - 3 } - \lambda ) (r _{ 24 } ^{ - 3 } - \lambda) \label{eqn:sigma3}. 
\end{align} 
\section{Uniqueness of Co-circular configurations}
In this section we want to prove Theorem \ref{thm:uniqueness}. We  break down the proof in four steps, which we summarize here.
\begin{enumerate} 
    \item
       We show that if $ \mathbf{r} \in \mathcal{M} ^{ + }   $ is a critical point of $ U|_{ \mathcal{M} ^{ + }  } $, then it is necessarily a nondegenerate local minimum. This is proved in Proposition \ref{prop:crit-points}. Lemma \ref{lem:lambda} is a technical lemma required to prove Proposition  \ref{prop:crit-points}.

    \item
       We study the topology of $ \mathcal{M}  $ and $ \mathcal{M} ^{ + } $. 
       In Lemma \ref{lem:M} we show that $ \mathcal{M} \approx S ^2 \times S ^2 $. In  Lemma \ref{lem:M+} we show that  $ \mathcal{M} ^{ + } $ is contractible and  the  Euler characteristic  $ \chi (\mathcal{M} ^{ + })   $  of $\mathcal{M} ^{ + } $ is $1$.  

    \item
       We use Morse theory to prove that the function  $U| _{ \mathcal{M} ^{ + } }  $ has a unique critical point  on $ \mathcal{M} ^{ + } $. This is done in  Lemma \ref{lem:uniqueness}. 
      \item We prove the theorem.
\end{enumerate} 
We start with the following technical lemma which is needed in  the proof of Proposition  \ref{prop:crit-points}.
 \begin{lemma} \label{lem:lambda}
 If $ \mathbf{r} ^\ast  \in \mathcal{M} ^{ + }   $  is a critical point of $ U| _{\mathcal{M} ^{ + }  } $ 
then $ \lambda >0 $. 
 \end{lemma} 
 \begin{proof}
     Suppose, for the sake of contradiction, that  $ \lambda \leq   0 $.  By the first of the two equation \eqref{eqn:cc_n1} we find that  
     \[\sigma \frac{ r _{ 34 } } { r _{ 12 } } = m _1 m _2 (r _{ 12 } ^{ - 3 } - \lambda) >0 \]
    and hence $ \sigma >0 $, since $ r _{ 12 }, r _{ 34 }  >0 $ in $ \mathcal{M} ^{ + } $.  
 By the first of the two equation \eqref{eqn:cc_n2} we find that  
     \[ -\sigma \frac{ r _{ 24 } } { r _{ 13 } }= m _1 m _3 (r _{ 13 } ^{ - 3 } - \lambda)  >0 \]
    and hence $ \sigma <0 $, which contradicts the fact that $ \lambda \leq 0  $. Hence, $ \lambda >0 $.   

 \end{proof} 

Note that the second derivative  $ D ^2 L (\mathbf{r} ; \lambda , \sigma) $  of $ L (\cdot ; \lambda , \sigma) $ with respect to the variable $\mathbf{r}$ is the matrix 
\[
    D ^2 L (\mathbf{r} ; \lambda , \sigma) = D ^2 U (\mathbf{r}) + \lambda M D ^2 I  (\mathbf{r}) + \sigma D ^2  P (\mathbf{r}).   
\]
If $ \mathbf{r} $ is a critical point of $ U | _{ \mathcal{M} ^{ + } } $ this second derivative, with appropriate choices of $ \lambda $  and $ \sigma $,  is the second derivative of  $ U | _{ \mathcal{M} ^{ + } } $, the restriction of $ U $ to $ \mathcal{M} ^{ + } $. We can now prove the following proposition

\begin{proposition} \label{prop:crit-points}
If $ \mathbf{r} ^\ast  \in \mathcal{M} ^{ + }   $  is a critical point of $ U| _{\mathcal{M} ^{ + }  } $ 
then $ \mathbf{r} ^\ast $ is  a nondegenerate minimum point for $ U|_{\mathcal{M} ^{ + } } $.
\end{proposition} 
\begin{proof} 
   The second derivative of $ L $ is the matrix
   \begin{align*} D ^2 L  (\mathbf{r} ; \lambda , \sigma) = & 
       \operatorname{diag}(f _{ 12 } (\mathbf{r}),f _{ 13 } (\mathbf{r}),f _{ 14 } (\mathbf{r}),
       f _{ 23 } (\mathbf{r}),f _{ 24 } (\mathbf{r}),f _{ 34 } (\mathbf{r})) \\
      &  + \operatorname{adiag } (\sigma ,- \sigma, \sigma, \sigma,- \sigma, \sigma)    
  \end{align*}
where $ f _{ ij } (\mathbf{r}) = m _i m _j (2 r _{ ij } ^{ - 3 } + \lambda) $. Here, 
$ \operatorname{diag } (f _{ 12 } ,f _{ 13 } ,f _{ 14 } ,
       f _{ 23 },f _{ 24 } ,f _{ 34 }) $ denotes the  $ 6 \times 6 $ diagonal matrix whose diagonal entries are $ f _{ 12 } ,f _{ 13 } ,f _{ 14 } ,
       f _{ 23 } ,f _{ 24 },f _{ 34 }$.
       Similarly, $ \operatorname{adiag } (\sigma ,- \sigma, \sigma, \sigma,- \sigma, \sigma) $ denotes the $ 6 \times 6 $ anti-diagonal matrix whose anti-diagonal   entries, starting from the upper right corner,  are  $ \sigma ,- \sigma, \sigma$, $\sigma,- \sigma, \sigma $.

Let $ \Delta _k (\mathbf{r})   $ be the  principal minor of order $ k $ of $  D ^2 L  (\mathbf{r} ; \lambda , \sigma) $  . We first prove that if $ \mathbf{r} ^\ast $ satisfies equations  (\ref{eqn:cc_n1}-\ref{eqn:cc_n3}), then  $ \Delta _k (\mathbf{r} ^\ast)  >0 $ for $ k = 1, \ldots 6 $. 
Since $ \lambda >0 $  by Lemma \ref{lem:lambda} it  is easy to see that the first three principal minors are always positive   
\begin{align*} 
    \Delta _1 ( \mathbf{r})   & =  \frac{{\left(\lambda r_{12}^{3} + 2\right)} m_{1} m_{2}}{r_{12}^{3}}>0\\
    \Delta _2 (\mathbf{r})  & = \frac{{\left(\lambda r_{12}^{3} + 2\right)} {\left(\lambda r_{13}^{3} + 2\right)} m_{1}^{2} m_{2} m_{3}}{r_{12}^{3} r_{13}^{3}}>0\\
    \Delta _3 (\mathbf{r})   & = \frac{{\left(\lambda r_{12}^{3} + 2\right)} {\left(\lambda r_{13}^{3} + 2\right)} {\left(\lambda r_{14}^{3} + 2\right)} m_{1}^{3} m_{2} m_{3} m_{4}}{r_{12}^{3} r_{13}^{3} r_{14}^{3}}>0.
\end{align*}
Let 

\begin{align*}
   A _0 (\mathbf{r})    & =m_{1} m_{2} m_{3} m_{4} ( \lambda^{2} r_{12}^{3} r_{34}^{3} + 2 \, \lambda r_{12}^{3} + 2 \, \lambda r_{34}^{3}  + 4 \, )- r_{12}^{3} r_{34}^{3} \sigma^{2}\\
   A _1 (\mathbf{r})   & = m_{1} m_{2} m_{3} m_{4} (\lambda^{2} r_{14}^{3} r_{23}^{3} + 2 \, \lambda r_{14}^{3} + 2 \, \lambda r_{23}^{3}  + 4 \,) - r_{14}^{3} r_{23}^{3} \sigma^{2}\\
   A _2 (\mathbf{r})  & = m_{1} m_{2} m_{3} m_{4} (\lambda^{2} r_{13}^{3} r_{24}^{3} + 2 \, \lambda r_{13}^{3} + 2 \, \lambda r_{24}^{3} + 4 \,) - r_{13}^{3} r_{24}^{3} \sigma^{2}
\end{align*}
then the remaining principal minors are:
\begin{align*}
    \Delta _4  (\mathbf{r})  & = \frac{ m _1  ^2 m _2m _3  } { r _{ 13 } ^3 r _{ 14 } ^3 r _{ 23 } ^3 } (2 r _{ 12 } ^{ - 3 } + \lambda)( \lambda r _{ 13 } ^3 + 2 ) A _1   \\
    \Delta _5 (\mathbf{r})   & = \frac{ m _1 m _2 } { r _{ 13 } ^3 r _{ 14 } ^3 r _{ 23 } ^3 r _{ 24 } ^3 } (2 r _{ 12 } ^{ - 3 } + \lambda) A _1 A _2 \\ 
    \Delta _6  (\mathbf{r})  & = \frac{ 1 } { r _{ 12 } ^3  r _{ 13 } ^3 r _{ 14 } ^3 r _{ 23 } ^3 r _{ 24 } ^3 r _{ 34 } ^3  } A _0 A _1 A _2.   
\end{align*} 
Since $ \lambda >0 $, eliminating $ \sigma ^2 $  from $ A _0 $ using  \eqref{eqn:sigma1}, from $ A _1 $ using  \eqref{eqn:sigma2}, and from $ A _2 $ using   \eqref{eqn:sigma3}, yields 
\begin{align*} 
    A _0 (\mathbf{r} ^\ast) = 3 m _1 m _2 m _3 m _4 (\lambda r _{ 12 } ^3 + \lambda r _{ 34 } ^3 + 1)>0\\ 
    A _1 (\mathbf{r} ^\ast) =3 m _1 m _2 m _3 m _4 (\lambda r _{ 14 } ^3 + \lambda r _{ 23 } ^3 + 1)>0   \\  
    A _2 (\mathbf{r} ^\ast) = 3 m _1 m _2 m _3 m _4 (\lambda r _{ 13 } ^3 + \lambda r _{ 24 } ^3 + 1)>0.
\end{align*} 
Consequently, we have that $ \Delta _k (\mathbf{r} ^\ast) >0 $ for $ k = 4,5,6 $.
Since all the principal minors are positive it follows that  $ D ^2 L(\mathbf{r} ^\ast , \lambda , \sigma ) $ is positive definite and $ \mathbf{r} ^\ast $  is a nondegenerate local minimum of $U| _{ \mathcal{M} ^{ + } } $. 
\end{proof} 

\begin{remark} 
    By Proposition \ref{prop:gradients}
 we see that the gradient of  $ U | _{ \mathcal{M} ^{ + } } $  and the gradient of $ U| _{ \mathcal{N} } $  coincide for any $ \mathbf{r} \in \mathcal{D} $. However,  the second derivative of  $ U | _{ \mathcal{M} ^{ + } } $ is in general different from the second derivative of $ U| _{ \mathcal{N} } $,  because the fact that  $ P $ and $ H $ are tangent at the critical points does not ensure that the quadratic approximation at those points is the same. Furthermore, the property concerning the gradients of $ U | _{ \mathcal{M} ^{ + } } $  and  $ U| _{ \mathcal{N} } $   given in Proposition \ref{prop:gradients} holds on $ \mathcal{D} $ and not on the larger set $ \mathcal{M} ^{ + } $. 
    Consequently, when looking at the second derivatives it is important to  be careful to  consider carefully the various restrictions of $ U $.  
\end{remark}

We now turn to study the topology of $ \mathcal{M} $.

\begin{lemma}\label{lem:M}
  $ \mathcal{M} \approx \operatorname{Gr} _{ + } (2, 4) \approx S ^2 \times S ^2 $.  
\end{lemma} 
\begin{proof}
   Consider the following change of coordinates: 
   \[ r _{ ij } =  \left( \frac{ 2 M   } {( m _i m _j) }\right) ^{ 1/2 }p _{ ij }  , \quad i, j \in {1,2,3,4} \mbox{ with } i< j. \] 
In these coordinates  the equations $I (\mathbf{r})  = 1 $  and $ P (\mathbf{r}) = 0 $ take the form  
\begin{align}
  p _{ 12 } ^2 + p _{ 13 } ^2 + p _{ 14 } ^2 + p _{ 23 } ^2 + p _{ 24 } ^2 + p _{ 34 } ^2  & = 1  \label{eqn:gras1}\\
     \left( \frac{ 2 M   } {   \sqrt{  m _1 m _2 m _3 m _4  }}\right) ( p _{ 12 } p _{ 34 } - p _{ 13 } p _{ 24 } + p _{ 14 } p _{ 23 })&  = 0\label{eqn:gras2}.
\end{align}
These are the  equations of   the oriented Grassmanian 
\[\operatorname{Gr} _{ + } (2, 4) = SO(4)/(SO (2) \times SO (2)   ), \] 
(see \cite{viro2004topology} for details). 
Equations (\ref{eqn:gras1}-\ref{eqn:gras2})   are equivalent to the system  
\begin{align} 
    (p _{ 12 } + p _{ 34 }) ^2 + (p _{ 13 } - p _{ 24 }) ^2 + (p _{ 14 } + p _{ 23 }) ^2 & = 1\label{eqn:spherep1} \\   
    (p _{ 12 } - p _{ 34 }) ^2 + (p _{ 13 } + p _{ 24 }) ^2 + (p _{ 14 } - p _{ 23 }) ^2 & = 1,\label{eqn:spherep2}
\end{align} 
which shows that $ \mathcal{M} $ is diffeomorphic to $ S ^2 \times S ^2 $. 
\end{proof}

We can  now better understand the topology of    $ \mathcal{M} ^{ + } $. 
Discussions with Shengda Hu were very helpful with this next lemma.  

\begin{lemma}\label{lem:M+}
 $\mathcal{M} ^{ + } $ is contractible and its Euler charactersitic $\chi(\mathcal{M}^{ + }) $ is $1$. 
\end{lemma}
\begin{proof}
After the change of variable  
    \begin{align*} 
    v _1 &  = p _{ 12 } + p _{ 34 }  & v  _2 & = p _{ 13 } - p _{ 24 } & v _3 & = p _{ 14 } + p _{ 23 } \\
    w _1 &  = p _{ 12 } - p _{ 34 }  & w _2  & = p _{ 13 } + p _{ 24 } & w _3 & = p _{ 14 } - p _{ 23 }.
\end{align*} 
equations \eqref{eqn:spherep1} and \eqref{eqn:spherep2} can be rewritten in the form    

    \begin{align*} 
        S _1 & = \{ v=(v _1 , v _2 , v _3) \in \mathbb{R}  ^3 : \, v _1 ^2 + v _2 ^2 + v _3 ^2 = 1 \},
        \\ \quad  S _2&  = \{ w= (w _1 , w _2 , w _3)   \in \mathbb{R}  ^3 : \, w _1 ^2 + w _2 ^2 + w _3 ^2 = 1\} .
    \end{align*} 
Clearly the set $ \mathcal{M} ^{ + } $  is  homeomorphic to $ E $,  the subset of  
$ S _1 \times S _2 $ defined by the following inequalities
\begin{align*}
    p _{ 12 } & = \frac{ v _1 + w _1 } { 2 }  \geq 0 & p _{ 13} & = \frac{ v _2 + w _2 } { 2 } \geq 0 &  p _{ 14 } & = \frac{ v _3 + w _3 } { 2 } \geq 0 \\ 
     p _{ 34 } & = \frac{ v _1 - w _1 } { 2 }  \geq 0 & p _{ 24 } & = \frac{ w _2 -v _2 } { 2 } \geq 0 &  p _{ 23} & = \frac{ v _3 - w _3 } { 2 } \geq 0.
\end{align*}
These inequalities can be  expressed in   a more compact form as 
\[
    v _1 \geq | w _1 |, \quad v _3 \geq |w _3 |, \quad w _2 \geq |v _2 |,
\]
and on  $ S _1 \times S _2 $ they  reduce to 
\begin{equation}\label{eqn:ineq} 
    v _1 \geq | w _1 |, \quad v _3 \geq |w _3 |, \quad w _2 \geq 0.
\end{equation} 
This can be explained as follows. The first two inequalities in \eqref{eqn:ineq} imply that $ |v _1 |^2 + |v _3 |^2 \geq |w _1| ^2 +| w _3| ^2 $, which gives 
\[
    |w _2 | = \sqrt{ 1 - |w _1| ^2 - |w _3| ^2 } \geq \sqrt{ 1 - |v _1 |^2 - |v _3 | ^2 } = |v _2 |.
\]
Hence $ w _2 > |v _2 | $ since $ w _2  \geq 0 $. 

The last inequality  in \eqref{eqn:ineq}, namely $ w _2 \geq 0 $,  selects the closed upper hemisphere $ H _2 ^{ + } $  of  $ S _2 $. The hemisphere $ H _2 ^{ + } $ is homeomorphic to a closed disk and  any point on it  can be represented with coordinates  $(w _1 , w _3)$ in $ \{ (w _1 , w _3) \in \mathbb{R}  ^2 | ~w _1 ^2 + w _3 ^2 \leq 1 \} $.

 Corresponding to each  point of coordinates $ (w _1 , w _3) $, there is a region   $F$ of the sphere $ S _1 $ determined by the inequalities  
\[v _1 \geq | w _1 |, \quad v _3 \geq |w _3 |.\] 
If $ (w _1,w _3)   =(\pm 1,0) $ then    $ (v_1,v_2,v_3)=(1,0,0) $. If $  (w _1 ,w _3) = (0,\pm 1) $ then $ (w _1, w _2 , w _3) = (0,0,1) $. Hence, in these cases the region $ F $ reduces to a point.  For any other value of $ (w _1 , w _3) $  the region $ F $ is homeomorphic to a closed 2-disk. 
The restriction of the  projection  $\tilde  p: (v _1 , v _2 , v _3, w _1 , w _2 , w _3) \to (w _1 , w _2 , w _3)$, induces a fibration  $ p :E \to H _2 ^{ + } $ with base space $ H _{ 2 } ^{ + } $ and fibers given by  $ F $. Thus, the projection $p$ is a fibration with contractible fibers. Since $H_2^+$ is also contractible, we see that the space $E$ is  contractible. In particular, $ \chi (\mathcal{M} ^{ + }) =  \chi(E) = 1$.
\end{proof} 

\begin{remark} 
The previous Lemma, and a much more general thorem,  seems to follow from a recent result by  Galashin, Karp, and  Lam \cite{galashin2017totally}.
Let  $ \operatorname{Gr}(k,n) $  denote the Grassmannian of $k$-planes in  
   $ \mathbb{R} ^n $, its  totally nonnegative part $ \operatorname{Gr} _{ \geq } (k, n) $ 
is defined to be the set of $  x \in \operatorname{Gr}(k, n)$  whose Pl\"ucker coordinates are all nonnegative. It has been shown that $  \operatorname{Gr} _{ \geq } (k, n)$ is homeomorphic to a $ k (n - k) $ dimensional  closed ball \cite{galashin2017totally}. Using this result it should be  possible to show that  the oriented Grassmmanian $ \operatorname{Gr} _{ + }  (n, k) $ with all the Pl\"ucker coordinates nonnegative, which we may call {\it totally nonnegative oriented Grassmannian},
 is a $ k (n - k) $ dimensional  closed ball. This would generalize the previous lemma to any oriented Grassmanian  $ \operatorname{Gr} _{ + }  (n, k) $. 
\end{remark}

Since we have determined the topology of $ \mathcal{M} ^{ + } $ we can now use Morse theory to prove the following Lemma
\begin{lemma} \label{lem:uniqueness}
    The function  $U| _{ \mathcal{M} ^{ + } }  $ has a unique critical point  on $ \mathcal{M} ^{ + } $. 
\end{lemma} 

\begin{proof}
   By Proposition \ref{prop:crit-points} any critical point 
    $ \mathbf{r} \in \mathcal{M} ^{ + } $ is a  nondegenerate local minimum of the function 
    $U| _{ \mathcal{M} ^{ + }} $, and hence $ U| _{ \mathcal{M} ^{ + }} $ is a Morse function that approaches $ + \infty $ as $\mathbf{r}  $ approaches $ \partial \mathcal{M} ^{ + } $, the boundary of $ \mathcal{M} ^{ + } $.
Therefore, the function $ U| _{ \mathcal{M} ^{ + } } $ admits a  global minimum value in the interior of $ \mathcal{M} ^{ + } $. 
Suppose there  are  several global  minimum points where the function obtains its least possible value.  By Proposition \ref{prop:crit-points} any of such point must be a  non-degenerate local minimum point. By Lemma \ref{lem:M+}, the Euler characteristic of  $ \mathcal{M} ^{ + } $ is $ \chi (\mathcal{M} ^{ + })  = 1 $.  By Morse theory we have 
\begin{equation}\label{eqn:Morse}
    1=\chi (\mathcal{M} ^{ + }) = \sum (- 1) ^{ \gamma } C ^{ \gamma }   
\end{equation} 
 where the sum is over the critical points,  $ \gamma $ is the Morse index of the critical points and $ C ^{ \gamma } $ is the number of critical points of index $ \gamma $. 
 We know that there is at least one local minimum, and that all the critical points of $ U | _{ \mathcal{M} ^{ + } } $ are local minimum points and hence   have index $ 0 $. However, this function cannot have more than one minimum point since otherwise,  equation \eqref{eqn:Morse} would imply  the existence of at least one non-minimum critical  point, contradicting  Proposition \ref{prop:crit-points}. 
\end{proof} 
We are now in a position to prove Theorem \ref{thm:uniqueness}, our main result 
\begin{proof}[Proof of Theorem \ref{thm:uniqueness}]
    Recall that, by Proposition \ref{prop:crit-points}, co-circular central configurations correspond to  distance vectors $ \mathbf{r} \in \mathcal{D} $  that are  critical points of the function $ U| _{ \mathcal{M} ^{ + }  } $. Lemma \ref{lem:uniqueness} shows that  $ U| _{ \mathcal{M} ^{ + }  } $ has a unique critical point on $ \mathcal{M} ^{ + }  $. Since $ \mathcal{D} \subset \mathcal{M} ^{ + } $, there is at most one critical point on $ \mathcal{D} $. Recall that  if 
    $ \mathbf{q} $ and $ \mathbf{q}' $ can be transformed one into the other with a reflection than they are mapped to the same distance vector $ \mathbf{r } $. Hence, we have shown that there is a most one equivalence class (with respect to the equivalence relation $\sim'$) of co-circular central configurations for each ordering of the masses, and the theorem follows.  
\end{proof}

\section*{Acknowledgments}
I would like to thank   Alain Albouy, Shengda Hu, Steven Karp,  Santiago L\'opez de Medrano, and  Alessandro Portaluri for interesting discussions on this work. This work was supported by an NSERC discovery grant. 
\bibliographystyle{plain}

\bibliography{references}

\end{document}